\documentclass[11pt,reqno]{amsart}
\usepackage{amsfonts,amssymb,amsmath,amsopn,amsthm,graphicx,dsfont}
\usepackage{amsxtra,mathrsfs}
\usepackage[mathscr]{eucal}
\usepackage{colordvi}
\usepackage[usenames,dvipsnames]{color}
\usepackage{mathtools}
\usepackage[colorlinks=true]{hyperref}

\usepackage{amsmath,amsfonts,amsthm,amssymb,amsxtra}
\usepackage{amsxtra, amssymb, mathrsfs}
\usepackage[usenames,dvipsnames]{color}

\newtheorem{theorem}{Theorem}[section]

\newtheorem{lemma}[theorem]{Lemma}

\theoremstyle{definition}

\theoremstyle{remark}
\newtheorem{remark}[theorem]{\bf Remark}

\numberwithin{equation}{section}

\newcommand{\pd}{\partial}
\newcommand{\eps}{\varepsilon}

\newcommand{\R}{\mathbb{R}}

\newcommand{\Z}{\mathbb{Z}}

\newcommand{\hh}{\mathscr{H}}

\begin{document}

\title[Heat kernel estimates for relativistic Hamiltonians with magnetic field]{Heat kernel estimates for two-dimensional
relativistic Hamiltonians with magnetic field}

\author{ Hynek Kova\v r\'{\i}k}

\address {Hynek Kova\v{r}\'{\i}k, DICATAM, Sezione di Matematica, Universit\`a degli studi di Brescia, Italy}

\email {hynek.kovarik@unibs.it}

\maketitle

\begin{center}
\emph{To Ari Laptev on the occasion of his 70th birthday.}
\end{center}


\begin{abstract}
We study semigroups generated by two-dimensional relativistic Hamiltonians with magnetic field. In particular, for compactly supported radial magnetic field we show
how the long time behaviour of the associated heat kernel depends on the flux of the field. Similar questions are addressed for Aharonov-Bohm type magnetic field.
\end{abstract}

\section{\bf Introduction} 
\label{sec-intro}
Consider a two-dimensional magnetic Laplacian formally given by 
\begin{equation} 
H=P^2,  \qquad P  = i\nabla +A
\end{equation}
in $L^2(\R^2)$, where $A: \R^2\to \R^2$ is a vector potential generating a magnetic field $B:\R^2\to \R$ through the relation rot$\, A=B$.  Recall that if $A\in L^2_{\rm loc}(\R^2)$, then $P^2$ is the unique self-adjoint operator associated with the closed quadratic form 
\begin{equation} 
Q[u] = \|(i\nabla +A) u\|_2^2, \qquad u\in D(Q), 
\end{equation}
with the form domain
$$
 D(Q) = \big\{u\in L^2(\R^2)\, : \ P u\in L^2(\R^2)\big\}\, .
$$
The main object of our interest in this paper is the integral kernel of the semigroup generated by the relativistic Hamiltonian 
\begin{equation} \label{rel-op}
\hh= \hh(A,m) = \sqrt{P^2 +m^2} -m, 
\end{equation}
where $m\geq 0$ is the mass of the particle. In particular, we are interested in the long time behaviour of $e^{-t \,\hh(A,m)}(x,y)$ and in its dependence on the magnetic field.  Note that for a massless particle in the absence of magnetic field, we have
\begin{equation}  \label{00}
e^{-t \,\hh(0,0)}(x,y) = e^{-t \sqrt{-\Delta} }(x,y) = \frac{t}{2\pi (t^2 +|x-y|^2)^{3/2} }\, ,  \qquad    x,y \in\R^2, \ \ t>0,
\end{equation}
see e.g.~\cite[Sec.~7.11]{LL}. Hence 
\begin{equation} 
e^{-t \,\hh(0,0)}(x,y)  \, \leq\, \frac{1}{2\pi t^2} \qquad t>0, 
\end{equation}
uniformly in $x$ and $y$. By the diamagnetic inequality this upper bound can be extended to $\hh(A,0)$; 
\begin{equation} \label{upprb-mag} 
\big |\, e^{-t\, \hh(A,0)}(x,y)\, \big |   \, \leq\, \frac{1}{2\pi t^2} \qquad t>0,
\end{equation}
see equation \eqref{diamag-2} below.

However, Laptev and Weidl proved in \cite{lw}, under mild regularity and decay assumptions on $B$, that the operator $H$ satisfies a Hardy-type inequality 
\begin{equation} \label{hardy} 
\int_{\R^2} | (i\nabla +A) u|^2 \ \geq \ \int_{\R^2} w\, |u|^2 \qquad \forall\, u\in D(Q),
\end{equation}
where $w\gneqq 0$ is a weight functions such that $w(x) = \mathcal{O}(|x|^{-2}$) as $|x|\to\infty$. Hence the presence of a magnetic field removes the singularity of the Green of $-\Delta$ at zero energy. 
This suggest that it should be possible to improve the decay rate in $t$ of the upper bound \eqref{upprb-mag}, probably at a cost of spatial weights. One of the results of this paper, Theorem \ref{thm-m0}, confirms this heuristic expectation for radial magnetic fields with compact support. 

In the massive case when $m>0$  the semigroup generated by $\hh(A,m)$ exhibits different behaviour for $t\to 0$ and for $t\to\infty$. However, for large times one  observes again a faster time decay of the associated heat kernel with respect to the heat kernel generated by the non-magnetic operator, see Theorem \ref{thm-mpos}. In Section \ref{sec-ab} we obtain analogous results  for the Aharonov-Bohm type magnetic fields.

\begin{remark}
It should be noted that the path integral methods developed in \cite{hil}  could possibly provide a tool for alternative proofs or even more general
results. 
\end{remark}

\section{\bf Radial magnetic field} 

\subsection{Preliminaries} In this section we will always assume that $B\in L^1(\R^2)$. Let
\begin{equation} 
\alpha = \frac{1}{2\pi}\!\! \int_{\R^2} B 
\end{equation}
be the total (normalized) magnetic flux, and let
\begin{equation} \label{kappa}
\kappa = \min_{m\in\Z} |m+\alpha| \in [0, 1/2]
\end{equation}
be the distance between $\alpha$ and the set of integers. Recall also that for any $A\in L^2_{\rm loc}(\R^2)$ the semigroup $e^{-t H}$ satisfies 
the diamagnetic inequality 
\begin{equation} \label{diamag}
\big |\, e^{-t H}(x,y)\, \big |   \, \leq\, \frac{1}{4\pi t} \, e^{-\frac{|x-y|^2}{4t}} \qquad \text{a.~e.} \ x,y\in\R^2 , \ \ t>0,
\end{equation}
see \cite{si,hsu, ka,ahs,hs}.

\subsection{\bf The case $m=0$}

\begin{theorem} \label{thm-m0}
Assume that $B$ is radial, continuous and compactly supported.  Then 
\begin{equation} \label{lpq}
e^{-t\, \hh }: \, L^p(\R^2)\to L^q(\R^2),  \qquad  1\leq p\leq q\leq \infty,
\end{equation}
and its kernel  satisfies
\begin{align} 
|\, e^{-t\,  \hh(A,0)}(x,y) | \, & \lesssim \, (1+|x|)^{\beta} (1+|y|)^{\beta}\ t^{-2-2\beta} , \qquad\qquad\qquad\qquad\quad \text{if} \  \kappa >0,   \label{upperb-non-int}\\
|\, e^{-t\,  \hh(A,0)}(x,y) | \, & \lesssim\, \log (2+|x|)^\theta  \log (2+|y|)^\theta \ t^{-2} \big[\log(2+t)]^{-2\theta}\qquad  \text{if} \ \kappa = 0,   \label{upperb-int}
\end{align}
for any $\beta\in[0,\kappa]$ and any $\theta\in[0,1]$ respectively. 
\end{theorem}

\begin{proof} 
We use the Poincar\' e gauge for the vector potential $A$;  
\begin{equation} 
A(x) = \frac{(-x_2, x_1)}{|x|^2}\, \int_0^{|x|} B(r)\,  r\,  dr, \qquad x\in\R^2.
\end{equation} 
Then rot$\, A=B$ and $A\in C^1(\R^2)\cap L^\infty(\R^2)$. In view of \cite[Thm.~6.1]{bhl} the semigroup
$$
e^{-t H }: \, L^p(\R^2)\to L^q(\R^2),  \qquad  1\leq p\leq q\leq \infty,
$$
is thus strongly continuous in $t>0$, and its integral kernel  $e^{-t H }(x,y)$ is jointly continuous in $(x,y,t)$.
On the other hand, inequality \eqref{diamag} shows that $ e^{-t H }$ is a contraction on $L^p(\R^2)$ for any $p\in[1,\infty]$.  Hence by \cite[Thms.~3.9.7, 4.3.1]{jac}  and \cite[Example 3.9.16]{jac} we have 
\begin{equation} \label{eq-jacob}
e^{-t \, \hh(A,0)} = e^{-t \, \sqrt{H}}  = \frac{t}{\sqrt{4\pi}} \int_0^\infty s^{-\frac 32}\, e^{-\frac{t^2}{4s}}\, e^{-s H}\, ds,
\end{equation}
with $e^{-t \, \hh(A,0)}$ being a strongly continuous contraction semigroup on $L^p(\R^2), \ p\in[1,\infty]$. Moreover, the above mentioned properties of  $e^{-t H }(x,y)$ in combination with \eqref{eq-jacob} imply that $e^{-t \, \hh(A,0)}$ admits an integral kernel $e^{-t\, \hh(A,0)}(x,y)$ which is jointly continuous in $(x,y,t)$. 
 Note also that
$$
\int_0^\infty s^{-\frac 52}\, e^{-\frac{t^2}{4s}}\,  ds = t^{-3}\!\!\int_0^\infty s^{-\frac 52}\, e^{-\frac{1}{4s}}\,  ds = 8 t^{-3}\!\int_0^\infty\!\sqrt{r} \ e^{-r}\,  dr = 8 t^{-3}\, \Gamma(3/2) = 4 t^{-3}\, \sqrt{\pi}\ .
$$
This in combination with \eqref{diamag} and \eqref{eq-jacob} yields
\begin{equation} \label{diamag-2}
\big |\, e^{-t\, \hh(A,0)}(x,y)\, \big |\, \leq \,  \frac{t}{(4\pi)^{3/2}} \int_0^\infty s^{-\frac 52}\, e^{-\frac{t^2}{4s}}\,  ds =  \frac{1}{2\pi  t^2}\, .
\end{equation}
Hence $ e^{-t \, \hh(A,0)}: L^1(\R^2)\to L^\infty(\R^2)$, and \eqref{lpq} follows. 

\smallskip

\noindent Let us now assume that $\kappa>0$. Then, by \cite[Thm.3.3.]{k11} 
\begin{equation} \label{eq-kov-1}
\big |\, e^{-t\, H}(x,y)\, \big |\, \leq \, C_0\, (1+|x|)^{\kappa}\, (1+|y|)^{\kappa}\ t^{-1-\kappa} \qquad x,y\in\R^2\ \ \ t>0
\end{equation}
holds for some $C_0$. Hence by \eqref{eq-jacob}   
\begin{align*}
|\, e^{-t\, \hh(A,0)}(x,y) | &\,  \lesssim\,  (1+|x|)^{\kappa}\, (1+|y|)^{\kappa} \ t \int_0^\infty s^{-\frac 52-\kappa}\, e^{-\frac{t^2}{4s}}\, \, ds \, \lesssim \, (1+|x|)^{\kappa}\, (1+|y|)^{\kappa}\ t^{-2-2\kappa} .
\end{align*}
In view of \eqref{diamag-2} this proves \eqref{upperb-non-int}.  In order to prove \eqref{upperb-int} we note that for $\kappa=0$, 
\begin{equation} \label{eq-kov-2}
e^{-t H }(x,x)\, \lesssim\,  \log (2+|x|) \log(2+|y|)\ t^{-1}  \big(\log(2+t))^{-2} \, ,
\end{equation}
see  \cite[Thms.~3.1, 3.2]{k11}. Thus, proceeding as above we get 
\begin{align}
&\qquad |\, e^{-t\, \hh(A,0)}(x,y) | \,  \lesssim\,    \log (2+|x|)  \log (2+|y|)  \ t \int_0^\infty s^{-\frac 52}\, \big(\log(2+s))^{-2} \, e^{-\frac{t^2}{4s}}\, \, ds\nonumber  \\
& \qquad\qquad =   \log (2+|x|)   \log (2+|y|)  \ t^{-2} \int_0^\infty r^{-\frac 32}\, \big(\log(2+r t^2))^{-2} \, e^{-\frac{1}{4r}}\, \, dr\nonumber \\
& \qquad\qquad \leq  \log (2+|x|)   \log (2+|y|)  \Big( t^{-2} \int_0^{t^{-1}} r^{-\frac 32}\, e^{-\frac{1}{4r}}\, dr + (\log(2+t))^{-2} \int_{t^{-1}}^\infty r^{-\frac 32}\, e^{-\frac{1}{4r}}\, dr \Big)\nonumber  \\
&\qquad \qquad  \lesssim  \log (2+|x|)  \log (2+|y|) \ t^{-2} \, \big[\log(2+t) \big]^{-2}\,   \label{int-kappa}
\end{align} 
where we have used the fact that
$$
(1+t^n)\,  \int_0^{t^{-1}} r^{-\frac 32}\, e^{-\frac{1}{4r}}\, dr = \mathcal{O}(1) \qquad \forall\, n\geq 1.
$$
Inequality 	\eqref{upperb-int}  thus follows from \eqref{diamag-2} and \eqref{int-kappa}.
\end{proof}

\begin{remark}
The faster time decay of $e^{-t\, \hh(A,0)}(x,y)$ with respect to $e^{-t\, \hh(0,0)}(x,y)$ is compensated by the spacial grow of $x$ and $y$, 
as expected. 
 \end{remark}

\medskip
\subsection{\bf The case $m>0$} For particles with positive mass we have

\begin{theorem} \label{thm-mpos}
Let $m>0$.  Under the assumptions of Theorem \ref{thm-m0} we have 
\begin{align} 
|\, e^{-t\,  \hh(A,m)}(x,y) | \, & \lesssim \, (1+|x|)^{\beta} (1+|y|)^{\beta}\ t^{-1-\beta} , \qquad\qquad\qquad\qquad\quad\ \text{if} \  \kappa >0,   \label{upperb-non-int-2}\\
|\, e^{-t\,  \hh(A,m)}(x,y) | \, & \lesssim\, \log (2+|x|)^\theta  \log (2+|y|)^\theta \ t^{-1} \big[\log(2+t)]^{-2\theta}\qquad  \text{if} \ \kappa = 0,   \label{upperb-int-2}
\end{align}
for $t\geq 1$, and any $\beta\in[0,\kappa]$ and any $\theta\in[0,1]$ respectively. 

\smallskip

\noindent Moreover, if $t\leq 1$, then 
\begin{align} 
|\, e^{-t\,  \hh(A,m)}(x,y) | \, & \lesssim \,  t^{-2}    \label{uppb-tsmall}
\end{align}
\end{theorem}

\medskip

\noindent The proof of Theorem \ref{thm-mpos} is based on the following technical result.

\begin{lemma} \label{lem-aux1}
For any $a>0$ there exist constants $C_1(a)$ and $C_2(a)$ such that 
\begin{equation} \label{eq-aux1}
\int_0^\infty r^a\, e^{-t\left(r-\frac{m}{2r}\right)^2}\, dr \, \leq \, C_1(a)\, m^{\frac a2}\, t^{-\frac 12} + C_2(a)\  t^{-\frac{1+a}{2}}
\end{equation} 
holds for all $t>0$.
\end{lemma}

\begin{proof}
Note that the function $s \mapsto s +\sqrt{s^2+2m}$ is an increasing bijection which maps $\R$ onto $(0,\infty)$. Hence we will apply the substitution 
\begin{equation} \label{eq-sub}
r = \frac{s +\sqrt{s^2+2m}}{2}\, ,
\end{equation}
and split the integration in \eqref{eq-aux1} in two parts as follows;
\begin{align*}
\int_0^{\sqrt{\frac m2}} r^a\, e^{-t\left(r-\frac{m}{2r}\right)^2}\, dr & = 2^{-1-a} \int_{-\infty}^0 \left(s +\sqrt{s^2+2m}\, \right)^a \left(1+\frac{s}{\sqrt{s^2+2m}}\right)\, e^{-ts^2}\, ds \\
& \leq 2^{-1-a} \int_{-\infty}^0 \left(s +\sqrt{s^2+2m}\, \right)^a \, e^{-ts^2}\, ds \\
& = 2^{-1-a} \int_{-\infty}^0 \left(\frac{2m}{\sqrt{s^2+2m} -s} \right)^a \, e^{-ts^2}\, ds \\
& \leq2^{-1-a}\, (2m)^{\frac a2}  \int_{-\infty}^0 e^{-ts^2}\, ds =  2^{-1-a} (2m)^{\frac a2}  \, \frac{\sqrt{\pi}}{2}\ t^{-\frac 12}\, ,
\end{align*}
and
\begin{align*}
\int_{\sqrt{\frac m2}}^\infty r^a\, e^{-t\left(r-\frac{m}{2r}\right)^2}\, dr & = 2^{-1-a} \int_0^{\infty} \left(s +\sqrt{s^2+2m}\, \right)^a \left(1+\frac{s}{\sqrt{s^2+2m}}\right)\, e^{-ts^2}\, ds \\
& \leq 2^{-a} \int_0^{\infty} \left(s +\sqrt{s^2+2m}\, \right)^a \, e^{-ts^2}\, ds \\
&\, \lesssim \int_0^{\infty}  \left(2s +\sqrt{2m}\, \right)^a\, e^{-ts^2}\, ds\,  \lesssim\,  \int_0^{\infty}  \left(s^a +m^{\frac a2}\right)\, e^{-ts^2}\, ds \\
&\, \lesssim    \ t^{-\frac{1+a}{2}} +  m^{\frac a2}\, t^{-\frac 12}\, .
\end{align*}
\end{proof}

\begin{proof}[\bf Proof  of Theorem \ref{thm-mpos}]
By \eqref{eq-jacob} 
\begin{align*}
 e^{-t\, \hh(A,m)}(x,y) &= e^{mt}\,  \frac{t}{\sqrt{4\pi}}  \int_0^\infty s^{-\frac 32}\, e^{-\frac{t^2}{4s}-m^2 s}\, e^{-s H }(x,y)\, ds \\
 & =  \frac{t}{\sqrt{4\pi}}  \int_0^\infty s^{-\frac 32}\ e^{-\left(\frac{t}{2\sqrt{s}}-m\sqrt{s} \right)^2}\, e^{-s H }(x,y)\, ds \, .
\end{align*}
Hence using \eqref{diamag} and the substitution 
\begin{equation} \label{subs} 
r= \frac{\sqrt{t}}{2\sqrt{s}} \, ,
\end{equation} 
we obtain from Lemma \ref{lem-aux1} 
\begin{align} \label{diamag-3} 
\big|\,  e^{-t\, \hh(A,m)}(x,y) \big | & \lesssim\,   t  \int_0^\infty s^{-\frac 52}\ e^{-\left(\frac{t}{2\sqrt{s}}-m\sqrt{s} \right)^2}\,  ds\  \lesssim\ t^{-\frac12}
 \int_0^\infty r^2\, e^{-t\, \left(r-\frac{m}{2r}\right)^2}\, dr \nonumber \\
 & \lesssim \ t^{-1} +t^{-2}.
\end{align}
This proves \eqref{uppb-tsmall}. Similarly, if $\kappa>0$ and $t\geq 1$, then using  \eqref{eq-kov-1}, \eqref{subs} and Lemma \ref{lem-aux1} we get  
\begin{align*}
\big|\,  e^{-t\, \hh(A,m)}(x,y) \big |&\, \lesssim\ (1+|x|)^{\kappa} (1+|y|)^{\kappa}\, \frac{t}{\sqrt{4\pi}}  \int_0^\infty s^{-1-\kappa}\ e^{-\left(\frac{t}{2\sqrt{s}}-m\sqrt{s} \right)^2} \, \frac{ds}{s^{3/2}} \,  \\
 & = \, \frac{2^{3+2\kappa}}{\sqrt{\pi}} \, (1+|x|)^{\kappa}\, (1+|y|)^{\kappa}\ t^{-\frac 12-\kappa} \int_0^\infty r^{2+2\kappa}\, e^{-t\, \left(r-\frac{m}{2r}\right)^2}\, dr \\
 & \lesssim \ (1+|x|)^{\kappa}\, (1+|y|)^{\kappa}\ t^{-1-\kappa}\, .
\end{align*}
The upper bound \eqref{upperb-non-int} then follows from Lemma \ref{lem-aux1} and \eqref{diamag-3}. In the same way we  deduce from \eqref{eq-kov-2} that if $t\geq 1$ and $\kappa=0$, then 
\begin{align*}
\big|\,  e^{-t\, \hh(A,m)}(x,y) \big |&\, \lesssim\   \log (2+|x|) \log (2+|y|) \ t^{-\frac 12} \int_0^\infty r^2\, \Big[ \log\big(2+\frac{t}{4r^2}\big)\Big]^{-2} e^{-t\, \left(r-\frac{m}{2r}\right)^2}\, dr \\
 & \lesssim \ \log (2+|x|) \log (2+|y|) \ t^{-\frac 12} \Big(\displaystyle\int_0^{t^{\frac 14}}  r^2 \Big[ \log\big(2+\frac{\sqrt{t}}{4}\big)\Big]^{-2} e^{-t\, \left(r-\frac{m}{2r}\right)^2}\, dr + \\
 & \qquad +\int_{t^{\frac 14}}^\infty  r^2\, e^{-t\, \left(r-\frac{m}{2r}\right)^2}\, dr\Big)  \\
 & \lesssim \ \log (2+|x|) \log (2+|y|) \ t^{-1} \big[ \log(2+t)\big ]^{-2} \, .
\end{align*}
To complete the proof it suffices to use \eqref{diamag-3} once more. 
\end{proof}

\begin{remark}
Note that in the massive case, contrary to the case $m=0$, the semigroup $e^{-t\, \hh(A,m)}(x,y)$ exhibits different behaviour for long respectively short times. Indeed, by \eqref{diamag-3} we have 
\begin{equation}
\| e^{-t\, \hh(A,m)} \|_{L^1 \to L^\infty}  = \mathcal{O}(t^{-2}) \quad  t \to 0, \qquad \| e^{-t\, \hh(A,m)} \|_{L^1 \to L^\infty}  = \mathcal{O}(t^{-1}) \quad t \to \infty .
\end{equation} 
This is caused by the different behaviour of the symbol $\sqrt{P^2 +m^2}\ -m$ for $P\to 0$ and for $P\to \infty$ respectively. Analogous effect occurs when $A=0$, see \cite[Sec.~7.11]{LL}. 
\end{remark}

\section{\bf Aharonov-Bohm type magnetic fields} 
\label{sec-ab}


\noindent In this section we consider the Aharonov-Bohm magnetic field in $\R^2$. The latter is generated by 
 the vector potential $A$ whose radial
and azimuthal components (in the polar coordinates) are given by
\begin{equation} \label{A-polar}
A(r,\theta) = (a_1(r,\theta),\, a_2(r,\theta)), \qquad a_1=0,\quad
a_2(r)= \left(0,\frac{\alpha}{r}\right) .
\end{equation}
where $\alpha$ is the magnetic. Note that $A\not\in L^2_{\rm loc}(\R^2)$. We will therefore proceed in a different way than in the previous section and define the semigroup $e^{-t \sqrt{H_\alpha}}$  with the help of the partial wave decomposition. We limit ourselves to the analysis of the masless case, $m=0$. The main result of section is stated  in Theorem \ref{thm-semigroup}.

\smallskip

\noindent We define the Hamiltonian $H_\alpha$ as the Friedrichs extension of $(-i\nabla + A)^2$  on
$C_0^\infty(\R^2\setminus\{0\})$. In other words, $H_\alpha$ is the self-adjoint operator in $L^2(\R^2)$ generated 
by the closure, in $L^2(\R^2)$, of the quadratic form
\begin{equation}
Q_\alpha[u] =  \int_0^{2\pi}\! \int_0^\infty \big(|\pd_r u|^2+r^{-2}\, |(-i\pd_\theta+\alpha)\, u|^2\big )\, r\, dr d\theta,
\quad u\in C_0^\infty((0,\infty)\times[0,2\pi)).
\end{equation}


\medskip

\subsection{\bf Partial wave decompostion} 
\noindent 
Given a function $f:\R^2\to \R$ we will often use the polar coordinate representation
$$
f(x,y) = f(r,r',\theta,\theta') \quad \Leftrightarrow \quad x=r(\cos\theta,
\sin\theta), \, \,  y= r'(\cos\theta', \sin\theta').
$$
By expanding a given function $u\in L^2(\R_+\times (0,2\pi))$ into a
Fourier series with respect to the basis $\{e^{i m\theta}\}_{m\in\Z}$ of $L^2((0,2\pi))$, we obtain
a direct sum decomposition
\begin{equation}
L^2(\R^2) = \sum_{m\in\Z} \oplus \, \mathcal{L}_m,
\end{equation}
where 
$$
\mathcal{L}_m= \left\{g\in L^2(\R^2)\, : \,  g(x)= f(r)\,
e^{i m\theta} \, a.e., \, \int_0^\infty |f(r)|^2\,  r\, dr<\infty\right\}. 
$$
Since the vector potential $A$ is radial, the
operator $H_\alpha$ can be decomposed accordingly to the direct sum
\begin{equation} \label{sum-gen}
H_\alpha =   \sum_{m\in\Z}  \oplus \left( h_m \otimes\mbox{id}\right)
\Pi_m,
\end{equation}
where $h_m$ are operators generated by the closures, in $L^2(\R_+, r
dr)$, of the quadratic forms
\begin{equation} \label{qm}
q_m [f] = \int_0^\infty\, \left(|f'|^2+\frac{(\alpha+m)^2}{r^2}\,
|f|^2\right)\, r\, dr
\end{equation}
defined initially on $C_0^\infty(0,\infty)$, and $\Pi_m:L^2(\R^2)\to
\mathcal{L}_m$ is the projector acting as
$$
(\Pi_m u)(r,\theta) = \frac{1}{2\pi}\, \int_0^{2\pi}\,
e^{im(\theta-\theta')}\, u(r,\theta')\, d\theta'.
$$
Consider now the operator
\begin{equation} \label{lbeta}
L_m  = U\, h_m\, U^{-1} \qquad \text{in \, \, \, } L^2(\R_+,
dr),
\end{equation}
where $U: L^2(\R_+, r\, dr)\to L^2(\R_+, dr)$ is the unitary mapping
acting as $(U f)(r) = r^{1/2} f(r)$. Note that $L_m $ is subject
to Dirichlet boundary condition at zero and that it coincides with
the Friedrichs extension of the differential operator
$$
-\frac{d^2}{dr^2}\, +\,  \frac{(m+\alpha)^2-\frac 14}{r^2}
$$
defined on $C_0^\infty(\R_+)$. From \cite[Sect. 5]{k11} we know that
\begin{equation} \label{diag}
W_m \, L_m \, W_m ^{-1}\, \varphi(p) = p\,
\varphi(p), \quad \varphi\in W_m (D(L_m )),
\end{equation}
where the mappings $W_m ,\, W_m ^{-1}: L^2(\R_+) \to L^2(\R_+)$ given
by
\begin{align*} 
(W_m \, u)(p) & = \int_0^\infty u(r)\sqrt{r}\,
J_{|m+\alpha|}(r\sqrt{p})\, dr, \\
(W_m ^{-1} \varphi)(r) & =
\frac 12\, \int_0^\infty \varphi(p)\sqrt{r}\,
J_{|m+\alpha|}(r\sqrt{p})\, dp
\end{align*}
extend to unitary operators
on $L^2(\R_+)$.

\subsection{The semigroup $e^{-t\, \sqrt{H_\alpha}}$}
By the spectral theorem 
\begin{equation} \label{hk-gen}
e^{-t\, \sqrt{H_\alpha}}= \sum_{m\in\Z}\, \oplus \left( e^{-t \sqrt{h_m}} \otimes\mbox{id}\right)
\Pi_m .
\end{equation}
Denote by $p_m(r,r',t)$ the integral kernel of $e^{-t \sqrt{h_m}}$ in $L^2(\R_+, r\, dr)$. From \eqref{lbeta} 
it follows that 
\begin{equation}
p_m(r,r',t) = \frac{1}{\sqrt{r r'}}\ e^{- t\,  \sqrt{L_m}}(r,r').
\end{equation}
On the other hand, in view of \eqref{diag} we get
\begin{align}
 \big( e^{- t\,  \sqrt{L_m} }\, g\big)(r) & =   \big( W_m ^{-1}\, e^{-t\sqrt{p}}\ W_m \, g\big)(r) \nonumber \\
& =\frac 12\, \int_0^\infty \sqrt{r r'}
\int_0^\infty e^{- t \sqrt{p}}\,   J_{|m+\alpha|}(r\sqrt{p})
J_{|m+\alpha|}(r'\sqrt{p})\,  dp\,  g(r')
\, dr' \, . \  \label{epsilon}
\end{align}
Hence
\begin{align} 
p_m(r,r',t) & =  \frac 12\, \int_0^\infty e^{- t \sqrt{p}}\,   J_{|m+\alpha|}(r\sqrt{p})\, 
J_{|m+\alpha|}(r'\sqrt{p})\,  dp \nonumber \\
& = \int_0^\infty e^{- t\, p}\,   J_{|m+\alpha|}(r\, p)\, 
J_{|m+\alpha|}(r'\, p)\, p\,  dp. \label{pm}
\end{align}
In order to simplify the notation we will use in the sequel the shorthand
\begin{equation} \label{z}
z := \frac{r^2}{t^2}\ ,
\end{equation}
From \cite[Eq.~4.14.(16)]{erd} we then get the explicit expression for $p_m$ on the diagonal;  
\begin{align}  
p_m(r,r,t) & = \int_0^\infty e^{- t\, p}\,   J^2_{|m+\alpha|}(r\, p)\,  p\,  dp \label{pm-diag} \\
& = \frac{4^{\nu}}{\pi}\,  (2 \nu+1)\ t^{-2}\, \left(\frac rt \right)^{2\nu}\, \frac{\Gamma^2(\nu+\frac 12)}{\Gamma(2\nu+1)}\ F\Big(\nu+\frac 12, \nu+\frac 32, 2\nu+1;\, -4z \Big),
\nonumber
\end{align}
where 
\begin{equation} \label{nu}
\nu =|m+\alpha|,
\end{equation}
and $F(a,b,c; w)$ denotes the Gauss hypergeometric series, see e.g.~\cite[Eq.~15.1.1]{as}. Using its integral representation 
\begin{equation} \label{F-int}
F(a,b,c\, ; w) = \frac{\Gamma(c)}{\Gamma(b)\, \Gamma(c-b)}\, \int_0^1 s^{b-1}\, (1-s)^{c-b-1}\,(1-s w)^{-a}\, ds, \quad {\rm Re}\, c > {\rm Re}\, b >0, 
\end{equation}
see \cite[Eq.~15.3.1]{as}, in combination with the transformation formula \cite[Eq.~15.3.5]{as}: 
\begin{equation}
F(a,b,c\, ; w) = (1-w)^{-b}\, F\Big(b,c-a,c\, ; \frac{w}{w-1} \Big  )
\end{equation}
we find that  
\begin{align} \label{pm-2} 
p_m(r,r,t) & =  \frac{2 \nu+1}{\pi t^2}\, (4z)^\nu\,  \frac{\Gamma^2(\nu+\frac 12)}{\Gamma(2\nu+1)} (1+4z)^{-\nu-\frac 32}\,  
F\Big(\nu+\frac 32, \nu+\frac 12, 2\nu+1;\, \frac{4z}{4z+1} \Big) \nonumber \\
&= \frac{2 \nu+1}{\pi t^2}\, (4z)^\nu\,  (1+4z)^{-\nu-\frac 32}\,   \int_0^1 s^{\nu-\frac 12}\, (1-s)^{\nu-\frac 12}\, \left(1-\frac{4z s}{4z+1}\right)^{-\nu-\frac 32}\, ds. \nonumber \\
&= \frac{2 \nu+1}{\pi t^2}\, (4z)^\nu\,   \int_0^1 s^{\nu-\frac 12}\, (1-s)^{\nu-\frac 12}\, \left(1+4z (1-s) \right)^{-\nu-\frac 32}\, ds \nonumber \\
&= \frac{2 \nu+1}{\pi t^2}\, (4z)^\nu\,   \int_0^1 s^{\nu-\frac 12}\, (1-s)^{\nu-\frac 12}\, \left(1+4z s \right)^{-\nu-\frac 32}\, ds.
\end{align}
We have 

\begin{lemma} \label{lem-pm}
There exists $\eps_0>0$ such that for every $\eps\in (0,\eps_0)$ there is $C_\eps>0$ for which the upper bound 
\begin{equation} \label{pm-upperb}
(1+r)^{-\frac 32-\eps}\, (1+r')^{-\frac 32-\eps}\ p_m(r,r', t) \ \leq \ C_\eps \ \frac{t^{-2-2\kappa}}{(|m+\alpha|+1)^{1+\eps}}
\end{equation}
holds for all $m\in\Z$ and all $t\geq 1$.
\end{lemma}

\begin{proof}
Put 
\begin{equation} \label{eps0}
\eps_0 := \min \Big\{ |m+\alpha| -\frac 32\, : \, m\in \Z \ \wedge \ |m+\alpha| > \frac 32\Big\} .
\end{equation}
Clearly we have $\eps_0>0$. Let $\eps <\eps_0$ and put $\rho= \frac 32 +\eps$. 
Keeping in mind the notation \eqref{nu} we will distinguish two cases depending on the value of $\nu$. 

\smallskip

\noindent Assume first that $\underline{\nu >3/2}$. 
In view of \eqref{z} and \eqref{pm-2} 
\begin{align*}
(1+r^2)^{-\rho}\, p_m(r,r,t) & =  \frac{4 \nu+2}{\pi t^2}\ \frac rt\,  (1+r^2)^{-\rho}\,  (4z)^{\nu-\frac 12}\,  \int_0^1 s^{\nu-\frac 12}\, (1-s)^{\nu-\frac 12}\, (1+4z s)^{-\nu-\frac 32}\, ds \\
& \lesssim \, t^{-3} \, \nu \, (4z)^{\nu-\rho}\, \int_0^1 s^{\nu-\frac 12}\, (1-s)^{\nu-\frac 12}\, (1+4z s)^{-\nu-\frac 32}\, ds \\
&  =\, t^{-3} \, \nu \, (4z)^{\nu-\rho}\, \int_0^1 s^{\nu-\frac 12}\, (1+4z s)^{\rho-\nu}\, (1-s)^{\nu-\frac 12}\, (1+4z s)^{-\rho-\frac 32}\, ds,
\end{align*}
where we have used the fact that $z\leq r^2$ by assumption. From \eqref{eps0} it follows that $\rho< \nu$. Hence 
\begin{align*}
(1+r^2)^{-\rho}\, p_m(r,r,t) 
& \lesssim \, t^{-3} \, \nu \, (4z)^{\nu-\rho}\, \int_0^1 s^{\nu-\frac 12}\, (4zs)^{\rho-\nu}\, (1-s)^{\nu-\frac 12}\,  ds \\
&= \, t^{-3} \, \nu   \int_0^1 s^{\rho-\frac 12}\,  (1-s)^{\nu-\frac 12}\, ds = t^{-3} \, \nu\, B\Big(\rho+\frac 12, \nu+\frac 12\Big) \\
& = \,  t^{-3} \, \nu\, \frac{\Gamma\left(\rho+\frac 12\right)\, \Gamma\left(\nu+\frac 12\right)}{\Gamma(\nu+\rho+1)}\, ,
\end{align*}
where $B(\cdot\, , \cdot)$ denotes the Euler beta function. Moreover, by the Stirling formula, see e.g.~\cite[Eq.~6.1.37]{as}, we have 
$$
\frac{\nu\, \Gamma\left(\nu+\frac 12\right)}{\Gamma(\nu+\rho+1)} \ \sim \ \nu^{-\rho+\frac 12} = \nu^{-1-\eps} \qquad \nu\to\infty.
$$
Therefore there exists a constant $C_1$ such that 
\begin{align} \label{pm-ub-1}
(1+r^2)^{-\rho}\, p_m(r,r,t)\,  & \leq  \ C_1\ t^{-3}\ \nu^{-1-\eps} \qquad \forall\ \nu >\frac 32.
\end{align}
Now let $\underline{0 \leq \nu \leq \frac 32}$. In this case we have $\nu \leq \rho$ and \eqref{pm-2} thus implies that 
\begin{align*} 
(1+r^2)^{-\rho}\, p_m(r,r,t)\,  & \leq  \frac{32}{\pi t^2}\ (1+r)^{-2\rho}\ \frac{r^{2\nu}}{t^{2\nu}} \, \leq \, \frac{32}{\pi}\ t^{-2-2\kappa},
\end{align*}
since $\nu \geq \kappa$ by definition. This together with \eqref{nu} and \eqref{pm-ub-1} gives
$$
(1+r^2)^{-\rho}\, p_m(r,r,t)\,  \leq  \ C_2\ t^{-2-2\kappa}\ (|m+\alpha |+1)^{-1-\eps} \qquad \forall\ t\geq 1 
$$
holds for all $\nu$ and some $C_2$. To complete the proof it suffices to use the semigroup property of $e^{-t\, \sqrt{h_m}}$, which implies that 
$$
p_m(r,r',t) \, \leq \, \sqrt{p_m(r,r,t)\, p_m(r',r',t) }\ .
$$
\end{proof}

\noindent As a consequence of Lemma \ref{lem-pm} and equations \eqref{hk-gen}, \eqref{pm-diag} we obtain 

\begin{theorem} \label{thm-semigroup}
Let $\eps_0$ be given by \eqref{eps0}. Then for every $\eps\in(0,\eps_0)$ there exists a constant $K_\eps$ such that
\begin{equation} 
\| (1+|x|)^{-\frac 32-\eps}\, e^{-t \sqrt{H_\alpha} }\, (1+|x|)^{-\frac 32-\eps} \, \|_{L^1(\R^2) \to L^\infty(\R^2)}\,  \leq \, K_\eps\ t^{-2-2\kappa}.
\end{equation}
Moreover, 
\begin{align*}
e^{-t \sqrt{H_\alpha} } (x,x) &= \frac{t^{-2}}{2\pi^2} \sum_{m\in\Z} \! 4^{\nu} (2 \nu+1)\,  \Big(\frac{|x|}{t} \Big)^{2\nu}\, \frac{\Gamma^2(\nu+\frac 12)}{\Gamma(2\nu+1)}\ F\Big(\nu+\frac 12, \nu+\frac 32, 2\nu+1;\, - \frac{4|x|^2}{t^2}\Big),
\end{align*}
where $\nu =|m+\alpha|$.
\end{theorem}

\medskip

\begin{remark}
The decay rate $t^{-2-2\kappa}$ in Theorem \ref{thm-semigroup} is sharp. This follows from  \eqref{hk-gen} and the asymptotic behavior of $p_m(r,r,t)$ as $t\to\infty$. Indeed, let $k\in\Z$ be such that $\kappa= |k+\alpha|$. Equation \eqref{pm-2} then implies that 
\begin{equation} 
\lim_{t\to\infty} t^{2+2\kappa}\, p_{k}(r,r,t)  = \frac{2\kappa+1}{\pi}\ (2r)^{2\kappa}\, B\Big(\kappa +\frac 12, \kappa+\frac 12\Big), 
\end{equation}

\smallskip

\noindent It should be also noted that if $\alpha\in\Z$, in which case $\kappa=0$, then the operator $H_\alpha$ is unitarily equivalent to the Laplacian in $L^2(\R^2)$ which satisfies 
\begin{equation}
\|\, e^{-t \sqrt{-\Delta} } \, \|_{L^1 \to L^\infty} = \frac{1}{2\pi t^2}\ ,
\end{equation}
for all $t>0$, see \eqref{00}.
\end{remark}

\begin{center}
{\bf Acknowledgments}
\end{center}
The author would like to thank Gabriele Grillo for useful discussions.

\bigskip


\end{document}